\documentclass{article}
\usepackage[IL2]{fontenc}
\usepackage[letterpaper,margin=1.00in]{geometry}
\usepackage{amsmath, amssymb, amsthm, amsfonts}
\usepackage{bbm}
\usepackage{amscd}
\usepackage{mathrsfs}

\usepackage{comment} 
\usepackage{ifthen}
\usepackage{tikz}
\usetikzlibrary{positioning,decorations.pathreplacing}
\usepackage{ bbold }
\usepackage{appendix}
\usepackage{graphicx}
\usepackage{color}
\usepackage{algorithm}
\usepackage[noend]{algpseudocode}
\usepackage{epstopdf}
\usepackage{wrapfig}
\usepackage{paralist}
\usepackage{wasysym}
\usepackage[textsize=tiny]{todonotes}

\usepackage{listings} 

\newcommand{\mtodo}[1]{\todo[color=blue!20]{#1}}

\usepackage{pdflscape} 

\usepackage{framed}
\usepackage[framemethod=tikz]{mdframed}
\usepackage[bottom]{footmisc}
\usepackage{enumitem}
\setitemize{noitemsep,topsep=3pt,parsep=3pt,partopsep=3pt}
\usepackage[font=small]{caption}
\usepackage{xspace}

\usepackage{thmtools} 
\usepackage{thm-restate} 

\usepackage{hyperref}
\hypersetup{
    unicode=false,          
    colorlinks=true,        
    linkcolor=red,          
    citecolor=darkgreen,        
    filecolor=magenta,      
    urlcolor=cyan           
}

\newtheorem{theorem}{Theorem}[section]
\newtheorem{lemma}[theorem]{Lemma}
\newtheorem{meta-theorem}[theorem]{Meta-Theorem}

\usepackage[capitalize, nameinlink,noabbrev]{cleveref}

\crefname{theorem}{Theorem}{Theorems}
\crefname{proposition}{Proposition}{Propositions}
\crefname{observation}{Observation}{Observations}
\crefname{lemma}{Lemma}{Lemmas}
\crefname{claim}{Claim}{Claims}
\crefname{problem}{Problem}{Problems}
\crefname{conjecture}{Conjecture}{Conjectures}
\crefname{question}{Question}{Questions}
\crefname{example}{Example}{Examples}
\crefname{fact}{Fact}{Facts}

\definecolor{darkgreen}{rgb}{0,0.5,0}

\usepackage{algcompatible}
\algnewcommand\algorithmicswitch{\textbf{switch}}
\algnewcommand\algorithmiccase{\textbf{case}}

\algdef{SE}[SWITCH]{Switch}{EndSwitch}[1]{\algorithmicswitch\ #1\ \algorithmicdo}{\algorithmicend\ \algorithmicswitch}%
\algdef{SE}[CASE]{Case}{EndCase}[1]{\algorithmiccase\ #1}{\algorithmicend\ \algorithmiccase}%
\algtext*{EndSwitch}%
\algtext*{EndCase}%

\newcommand{\eps}{\varepsilon}

\renewcommand{\P}{\textrm{P}}

\newcommand{\poly}{\operatorname{poly}}

\renewcommand{\phi}{\varphi}

\renewcommand{\paragraph}[1]{\vspace{0.15cm}\noindent {\bf #1}:}

\newcommand{\FullOrShort}{full}
\ifthenelse{\equal{\FullOrShort}{full}}{
  
  \newcommand{\fullOnly}[1]{#1}
  \newcommand{\shortOnly}[1]{}
  \renewcommand{\baselinestretch}{1.00}
  
  }{
    \renewcommand{\baselinestretch}{1.00}
    \newcommand{\fullOnly}[1]{}
    \newcommand{\IncludePictures}[1]{}
   
  }

\title{Dynamic O(arboricity) coloring in polylogarithmic worst-case time }

\begin{document}
\date{}
 \author{
   Mohsen Ghaffari \\
   \small{MIT}\\
   \small{ghaffari@mit.edu}
   \and
   Christoph Grunau \\
   \small{ETH Zurich}\\   \small{cgrunau@inf.ethz.ch}
}

\maketitle
\begin{abstract}
A recent work by Christiansen, Nowicki, and Rotenberg [STOC'23] provides dynamic algorithms for coloring sparse graphs, concretely as a function of the graph's arboricity $\alpha$. They give two randomized algorithms: $O(\alpha \log \alpha)$ implicit coloring in $\poly(\log n)$ worst-case update and query times, and $O(\min\{\alpha \log \alpha, \alpha \log\log\log n\})$ implicit coloring in $\poly(\log n)$ amortized update and query times (against an oblivious adversary). We improve these results in terms of the number of colors and the time guarantee: First, we present an extremely simple algorithm that computes an $O(\alpha)$-implicit coloring with $\poly(\log n)$ amortized update and query times. Second, and as the main technical contribution of our work, we show that the time complexity guarantee can be strengthened from amortized to worst-case. That is, we give a dynamic algorithm for implicit $O(\alpha)$-coloring with $\poly(\log n)$ worst-case update and query times (against an oblivious adversary).
\end{abstract}
\thispagestyle{empty}
\newpage
\setcounter{page}{1}

\section{Introduction}
This paper is centered on dynamic algorithms for arboricity-dependent graph coloring. We present a randomized dynamic implicit coloring algorithm with polylogarithmic worst-case update and query times and using a number of colors within a constant factor of the optimal. Next, we first set the context and review the state of the art. Then, we overview our results.  

\subsection{Background: Arboricity-Dependent Coloring and Dynamic Algorithms}
\noindent\textbf{Arboricity.} The arboricity $\alpha(G)$ of a graph $G=(V, E)$ is a measure of its density in the densest part. By one definition, $\alpha(G)$ is the minimum number of forests to which the edges of $G$ can be decomposed. An alternative definition, which was shown to be equivalent by Nash-Williams~\cite{nash1964decomposition}, is $\alpha(G)=\max_{S\subseteq V} \lceil\frac{|E[S]|}{|S|-1}\rceil$. That is, roughly speaking, $\alpha_{G}$ is the maximum edge density among subgraphs of $G$. 

Arboricity is within a constant factor of numerous other standard measures of sparsity, including degeneracy, core number, width, linkage, coloring number, and Szekeres–Wilf number. For instance, the degeneracy $g(G)$ of a graph $G$ is the smallest value of $g$ such that the edges can be oriented to form a directed acyclic graph with outdegree at most $g$. One can see that $\alpha(G) \leq g(G)\leq 2\alpha(G)$.

\paragraph{Arboricity-Dependent Coloring} Graph coloring is hard even for modest approximations: for all $\eps >0$, approximating the chromatic number within $n^{1-\eps}$ is NP-hard~\cite{zuckerman2006linear}. Much of the algorithmic work on coloring focuses on coloring as a function of simpler properties of the graph. One example is $\Delta+1$ coloring, where $\Delta$ denotes the maximum degree. This is a trivial sequential problem but it has been widely studied in many modern computational settings including distributed, parallel, streaming, dynamic models, and sublinear-time centralized, see e.g.,~\cite{luby1985simple,assadi2019sublinear,chang2018optimal,chang2019complexity, ghaffari2022deterministic, bhattacharya2022fully}.

With $\Delta+1$ coloring, even a single high-degree node can blow up the number of colors. Arboricity gives a much more refined and robust measure.\footnote{The addition of any set of $k$ vertices and arbitrary edges incident on them increases the arboricity by at most $k$.} Any graph with arboricity $\alpha$ admits a vertex coloring with $2\alpha$ colors and this bound is sharp. For the former, notice that we can iteratively remove a vertex of degree at most $2\alpha-1$ until we exhaust the vertices (there is always such a vertex to remove). Then, we color the vertices greedily, in reverse of this order. The bound's sharpness is showcased by a clique of size $2\alpha$. Notice also that $\alpha(G)\leq \Delta$. In this paper, our focus is on dynamic algorithms that compute an arboricity-dependent coloring with an asymptotically optimal number of colors.   

\paragraph{Dynamic Algorithms}
Often, we need to solve the computational problem for a long list of inputs, where each input involves small changes to the previous input. This can arise naturally in practical settings but it is also useful in algorithmic applications of one subroutine in a larger algorithm. In such scenarios, ideally, we want not to resolve the problem from scratch; the algorithm should instead smoothly and quickly adapt to the changes. Our particular focus is on dynamic algorithms for coloring. Here, the graph undergoes updates in the format of edge insertions or deletions, and we want to have a coloring of the graph resulting from the updates. A key parameter of interest, known as the \textit{update time}, is the time needed to process each update.

\paragraph{Explicit vs. Implicit Outputs}
There is a distinction in how the dynamic algorithm provides access to the output. In an \textit{explicit dynamic algorithm}, with each update, the algorithm updates the output. However, in some problems, directly updating the output might be expensive, because even a single change in the input can require a large number of changes in the output---this is known as a \textit{high recourse}. For instance, it is known that when maintaining $2$-edge-connected components, or planar embeddings, even a single change can necessitate up to $\Theta(n)$ changes in the output. There is a workaround definition for these scenarios that can give something essentially as useful as directly updating the outputs: we prefer that, instead of directly adjusting the output which requires a lot of time, the dynamic algorithm processes each update quickly and adjusts a data structure. This data structure then gives fast access to any part of the output. This way, for every query to any single part of the output, we can determine that part quickly by using the adjusted data structure, in a time known as the \textit{query time}. 

It is known that explicit arboricity-dependent coloring requires a high update time: Barba et al.~\cite{barba2019dynamic} showed that any algorithm for explicit $f(\alpha)$ coloring requires update time $n^{\Omega(1)}$. Concretely, even $O(1)$ edge updates can necessitate $n^{\Omega(1)}$ color changes. Because of this, the focus (in prior work as well as in our work) has moved to implicit arboricity-dependent dynamic colorings with small update and query times. 

\subsection{State of the Art}
\paragraph{Orientation and Coloring} Arboricity-dependent coloring algorithms, including dynamic versions, usually involve two ingredients: (1) an algorithm that computes an orientation of the edges such that each node has a small outdegree, and (2) a coloring scheme that uses the orientation. The orientation is useful because we can make each node responsible for avoiding color conflicts with only its out-neighbors. For instance, in static computations, given an acyclic orientation with outdegree at most $2\alpha-1$, we can color the nodes greedily, by iteratively coloring a sink node using a color in $\{1, 2, \dots, 2\alpha\}$ and then removing it from the graph. The orientation ingredient is relatively well-understood in the dynamic setting: there are explicit dynamic algorithms that provide an orientation with outdegree at most $d=O(\alpha)$ in $\poly(\log n)$ worst-case update time~\cite{christiansen2022adaptive}. But the second ingredient still has much room for improvement. Our work, as well as the work of Christiansen, Nowicki, and Rotenberg~\cite{Christiansen23}, is on implicit dynamic coloring with a small number of colors, given a dynamically-maintained orientation with outdegree at most $d=O(\alpha)$.

\paragraph{Direct implications of classic distributed coloring algorithms} As pointed out by Henzinger et al.~\cite{henzinger2020explicit} and Christiansen et al.~\cite{Christiansen23}, given the orientations, one can obtain certain implicit colorings by importing classic distributed/local algorithms, essentially in a black-box manner. Let us review this.

Henzinger et al.~\cite{henzinger2020explicit} implicitly import a $2^{O(\alpha)}$ coloring approach of Goldberg, Plotkin, and Shannon~\cite{goldberg1987parallel}. Consider the orientation with outdegree $d$, and decompose the graph into $d$ edge-disjoint graphs $F_1, F_2, \dots, F_d$ of out-degree one. For that, simply let each node put one of its outgoing edges in each $F_i$. It is known by a classic distributed algorithm of Cole and Vishkin~\cite{cole1986deterministic} that we can compute a $3$-coloring of each $F_i$ in a simple local fashion\footnote{Henzinger et al. instead compute a $2$-coloring of each forest $F_i$, by dynamically maintaining the distances from a root and using the parity as a color. However, for that, they use $O(\log n)$ update time for each $F_i$. What we describe gives a faster dynamic algorithm and the degradation from $2$ to $3$ does not invalidate the $2^{O(\alpha)}$ number of colors.}: each node's color can be determined as a simple deterministic function of the $O(\log^* n)$ closest nodes in its outreach, i.e., nodes reachable by following the outgoing edges up to distance $O(\log^*n)$. This gives a $3$ coloring for each $F_i$. The product of these colors gives a $3^{d}=2^{O(\alpha)}$ coloring of the entire graph. Once the algorithm is queried for the color of a node $v$, the implicit dynamic coloring algorithm does as follows: learn the $O(\log^*n)$ outreach of $v$ along each of the subgraphs $F_i$, in a total of $O(\alpha \log^* n)$ time, and apply Cole-Vishkin to determine the color of $v$ in $\{1, 2, \dots, 3^{d}\}$. 

Christiansen et al.~\cite{Christiansen23} mention that one can improve the number of colors nearly exponentially, by using a distributed algorithm of Linial~\cite{linial1992locality}. Linial gives a very local recoloring scheme that reduces the number of colors rapidly: Given a current $k$-coloring, each node can compute its own $k'$-coloring for $k'=O(d^2 (\log_{d}{k})^2)$, by simply knowing the current colors of its up to $d$ outneighbors. We can apply this scheme twice on top of the $3^{d}$ coloring mentioned above. The first application reduces the number of colors to at most $O(d^4)$. The second application reduces the number of colors to $O(d^2)$. This gives an implicit dynamic algorithm for $O(d^2)=O(\alpha^2)$ coloring. Let us discuss the query time: to determine the color of each queried node $v$, we need to compute the $3^{d}$ coloring of each of outneighbors of $v$ and also each of the outneighbors of those outneighbors. Then we can compute the $O(d^4)$ colorings of the immediate outneighbors of $v$. And that lets us compute the $O(d^2)$ coloring of node $v$. Hence, the query time is $O(d^2) \cdot O(d\log^* n) \in \poly(d\log n)= \poly(\alpha\log n)$.

\paragraph{Improved Colorings of Christiansen et al.~\cite{Christiansen23}} Besides mentioning the above $O(\alpha^2)$ coloring that follows from classic distributed algorithms, as the main technical contributions of their work, Christiansen et al.~\cite{Christiansen23} provide two dynamic algorithms:

\begin{itemize}
\item[-] A randomized implicit coloring with $O(\alpha \log \alpha)$ colors and worst-case update and query times $\poly(\log n)$.
\item[-] A randomized implicit coloring with $O(\min\{\alpha \log \alpha, \alpha \log\log\log n\})$ colors and \textit{amortized} update and query times $\poly(\log n)$.
\end{itemize}

\subsection{Our Results}
Our main contribution is a randomized implicit coloring algorithm that improves the above results, by tightening the number of colors to the \textit{asymptotically optimal} bound $O(\alpha)$, with \textit{worst-case} $\poly(\log n)$ time guarantees: 

\begin{theorem}\label{thm:main} There is a randomized dynamic algorithm that computes an $O(\alpha)$ implicit coloring in $\poly(\log n)$ worst-case update and query times, against an oblivious adversary.
\end{theorem}
Furthermore, along the way to this result, we present an extremely simple randomized algorithm that computes an $O(\alpha)$ implicit coloring in $\poly(\log n)$ update time and with $\poly(\log n)$ amortized query time. The deterministic core of this amortized algorithm can be stated as follows.

\begin{theorem}\label{thm:amortized} There is a deterministic dynamic algorithm that computes an $O(\alpha)$ implicit coloring with $\poly(\log n)$ worst-case update time and $O(\alpha)$ amortized query time. 
\end{theorem}

\section{Preliminaries}
We use two subroutines from prior work: One is for computing a low-outdegree orientation in the dynamic setting. The other is the random vertex partitioning, which shows that, for the randomized coloring algorithm, it suffices to work on graphs with arboricity at most $O(\log n)$. We next state these. 

\begin{theorem} [Christiansen, Holm, Van der Hoog, Rotenberg, Schwiegelshohn~\cite{christiansen2022adaptive}]
\label{thm:orientation}
There is a deterministic explicit dynamic algorithm that computes an orientation with out-degree at most $d=O(\alpha)$, with update time $\poly(\log n)$. Here, $\alpha$ is the arboricity of the graph getting oriented at each time.    
\end{theorem}

The following is a streamlined version of Lemma 2.5 of Christiansen, Nowicki, and Rotenberg~\cite{Christiansen23}, which shows that, for query time purposes, it suffices to solve the coloring on graphs with arboricity at most $O(\log n)$. The intuitive reason is that one can use a random partitioning of the vertices into $\alpha/\log n$ parts, which are colored with separated colors, and each part induces a graph with arboricity $O(\log n)$. Please see ~\cite[Lemma 2.5]{Christiansen23} for details.
\begin{theorem} [Christiansen, Nowicki, and Rotenberg~\cite{Christiansen23}]
\label{thm:partition}
Assume that there exists a dynamic algorithm $\mathcal{A}$ that maintains an $O(\alpha)$ implicit coloring of the graph in $\poly(\log n)$ update time and $\poly(\alpha \log n)$ query time, against an oblivious adversary and with high probability, where $\alpha$ is the arboricity of the current graph. Then, there exists a dynamic algorithm $\mathcal{B}$ that maintains an implicit $O(\alpha)$ coloring in $\poly(\log n)$ update time and $\poly(\log n)$ query time, against an oblivious adversary and with high probability, where $\alpha$ is the arboricity of the current graph.   
\end{theorem}

\section{Our Algorithms}
\label{sec:Algorithm}
We first overview the general setup, which is common to all our algorithms. Our dynamic color algorithms make use of known schemes for low-outdegree orientation. Our novelty is in computing a coloring, given these orientations. The general setup in our implicit coloring algorithms is as follows: we assume that the algorithm has access to an orientation with outdegree at most $d$, for some $d= O(\alpha)$ and $d\geq 2$, which is maintained dynamically as the graph undergoes changes. Concretely, the work of Christiansen et. al. provides a dynamic orientation with such an outdegree, which is maintained with $\poly(\log n)$ worst-case update time~\cite{christiansen2022adaptive}. We denote the resulting directed graph by $G=(V, A)$, and we concretely assume that for each node we have the list of its up to $d$ outgoing edges. Having access to this low outdegree orientation, the coloring algorithm is given a sequence of queries in an online fashion. Each query consists of a single node in $G$. The algorithm answers the query by assigning the queried node $v$ some color from the set $[O(\alpha)]$. The color assignments are such that if one queries all the nodes in an arbitrary order, then the resulting coloring is a proper coloring in the sense that no two neighboring nodes have the same color.

The rest of this section consists of three subsections. The first subsection describes our algorithm with an amortized query time. The second subsection provides an intuitive discussion of our algorithm that achieves a worst-case query time. The third subsection presents the actual algorithm and its analysis.

\subsection{Coloring with Amortized Query Time}
Next, we describe an extremely simple implicit coloring algorithm that given an orientation with outdegree $d$, for some $d= O(\alpha)$ and $d\geq 2$, computes an $O(\alpha)$ coloring with amortized query time $O(\alpha)$. Adding the $\poly(\log n)$ worst-case update time of computing an orientation~\cite{christiansen2022adaptive}(\Cref{thm:orientation}), this gives a dynamic algorithm for $O(\alpha)$ implicit coloring with $\poly(\log n)$ update time and $O(\alpha)$ amortized query time, thus proving \Cref{thm:amortized}. Furthermore, by adding the known random partitioning idea (\Cref{thm:partition}), which effectively brings down the arboricity to $O(\log n)$, one can get an $O(\alpha)$ implicit coloring algorithm against an oblivious adversary, with $\poly(\log n)$ update time and $\poly(\log n)$ amortized query time.

\paragraph{Algorithm Framework} Let $V^{colored} \subseteq V(G)$ denote the subset of vertices that already received their permanent colors when we processed the previous queries. In particular, we assume that all previously queried nodes (before any edge updates) are contained in $V^{colored}$; however, $V^{colored}$ can contain some extra vertices. We maintain the invariant that, after finishing every query, for each uncolored node, the total number of its incoming neighbors that are in $V^{colored}$ is at most $6d$.\footnote{We could set this $6d$ value to be $1.001d$ for this amortized algorithm, and the analysis would still work. We choose the value $6d$ simply to keep the presentation consistent with that of our worst-case algorithm, which is discussed in the next section.} Also, for each node (which is not in $V^{colored}$), we keep the list of its at most $6d$ incoming neighbors that are already colored.

Now, we describe how we process the next query at node $u$. If node $u$ is already in  $V^{colored}$, we simply output its computed color. If not, we first compute a subset $V^*$ of vertices to be colored, including $u$, and then we compute a coloring of all nodes of $V^*$, therefore adding them to $V^{colored}$. We compute the set $V^*$ in a manner that allows us to keep the invariant that if a node $v$ did not get a permanently assigned color during previous queries, then at most $6d$ of its in-neighbors got assigned a color. In particular, we maintain an invariant that for any $v\notin V^{colored}\cup V^*$, the number of its in-neighbors in $V^{colored}\cup V^*$ is at most $6d$.

Besides $V^*$, the algorithm also maintains a set $A^p \subseteq A$ of \textit{processed arcs}. We reset this set to empty when we make the first color query after an edge update (and reorientation)\footnote{For now, think that we take a fresh empty set $A^{p}$ with every edge change that results in a reorientation. This might seem to cause a memory that grows with the number of edge updates, but one can cut that to memory asymptotically linear in the number of edges in the graph.}. We keep an additional invariant that an arc $(v,w)$ is in $A^p$ if and only if $v$ is in $V^{colored} \cup V^*$. Each node that is not in $V^{colored} \cup V^*$ keeps a list of its incoming arcs in $A^{p}$ (as well as the size of this list).

Next, we first discuss how we color $V^*$ (this helps also in understanding the relevance of this invariant), and then we explain how we compute $V^*.$



\smallskip
\paragraph{How do we color $V^*$?} Suppose a set $V^*\subseteq (V\setminus V^{colored})$ has already been specified and our task is to compute a coloring of all vertices in $V^*$. In particular, from the invariants about the previous queries, we know that each node in $V^*$ has at most $6d$ in-neighbors in $V^{colored}$. 
First, we identify all the edges of the subgraph $G[V^*]$ induced by $V^*$ in $O(|V^{*}| d)$ time, by examining the outgoing edges of $V^*$ and keeping only edges that end in another node in $V^*$. Now, notice that this subgraph $G[V^*]$ has arboricity at most $\alpha$ and therefore degeneracy at most $2\alpha - 1$. Thus, we can order the vertices in $V^*$ in such a way that each node has at most $2\alpha - 1$ neighbors that are after it in the order. For that, we use a known static algorithm of Matula and Beck~\cite{matula1983smallest}, which runs in time $O(|V^*| + |E[V^*]|)$. Then, we color the vertices in reverse of this order, each time taking into account the at most $6d+d+(2\alpha-1)<9d$ neighbors that have been colored before. Here, $6d$ comes from incoming neighbors outside $V^*$ which were colored during the processing of previous queries, $d$ comes from outgoing neighbors outside $V^*$ colored before, and $2\alpha-1$ comes from the neighbors in $V^*$ appearing after the node in the order of $V^*$ and thus colored earlier.

\smallskip
\paragraph{How do we compute $V^*$}
The set $V^*$ is computed by a recursive algorithm. The initial input is the node $u$ which has been queried for its color. In general, the input to the recursive algorithm is a node $v$ that is not in $V^{colored}\cup V^*$.


Whenever we have a recursive call on node $v\notin V^{colored}\cup V^*$, we do as follows. First, we add $v$ to $V^*$. Then, we iterate through the outgoing arcs of $v$. For each arc $(v,w)$ outgoing from $v$ to a node $w$, we do as follows. We first add $(v, w)$ to $A^{p}$. If $w$ is in $V^{colored}\cup V^{*}$, we are done with processing this arc. Suppose that $w$ is not in $V^{colored}\cup V^{*}$. Then, we first add $(v, w)$ to the list of processed arcs incoming to $w$. Next, if $|A^{p}\cap IN(w)|\geq 6d$, where $IN(w)$ is the set of arcs incoming to $w$, then we invoke a recursion on node $w$. 

Notice that, modulo this recursive call, processing each arc is done in $O(1)$ time. Therefore, the overall complexity of the recursive algorithm is $O(|V^*|d)$. Also, note that, after the recursive algorithm has finished, we have the invariant that an arc $(v,w)$ is in $A^p$ if and only if $v$ is in $V^{colored} \cup V^*$, and each node that is not in $V^{colored} \cup V^*$ has a list of its incoming arcs in $A^{p}$.

\begin{lemma}\label{lem:amortized} The above algorithm computes a correct coloring, with amortized query time $O(d)$.
\end{lemma}

\begin{proof} Note that the algorithm terminates, because, with each recursion on a node $v$, we add $v$ to $V^*$ and therefore never recurse on $v$ again (as recursion is invoked only on nodes not in $V^{colored}\cup V^*$). For correctness, it just remains to argue that after the recursive algorithm terminated, for any $v\notin V^{colored}\cup V^*$, the number of its in-neighbors in $V^{colored}\cup V^*$ is at most $6d$. For the sake of contradiction, suppose the number is at least $6d+1$. Since we had the invariant from processing previous queries that $v$ had at most $6d$ neighbors in $V^{colored}$, there must have been at least one arc incoming to $v$ that got processed during this color query. Let $(v', v)$ be the last such arc that got processed. Then, when we read the size$|A^{p}\cap IN(v)|$, we would see that it is at least $6d+1$ and we would trigger the recursion on $v$, hence it gets added to $V^*$, contradicting the assumption that $v\notin V^{colored}\cup V^*$.

It remains to argue that the amortized query time is $O(d)$. Notice that, as argued above, for each color query, computing the set $V^*$ and coloring it takes $O(|V^*|d)$ time. Hence, it suffices to show that, starting with an initially uncolored graph, for any $i\in \mathbb{N}$, querying $i$ nodes for their colors results in at most $O(i)$ nodes that got added to $V^*$ and then to $V^{colored}$. 

Let us use as a potential the number of arcs $(v,w) \in A^p$ where $v\in V^{*}\cup V^{colored}$ and $w\notin V^{*}\cup V^{colored}$.\footnote{We note that a similar potential was used in the analysis of the randomized $O(\min\{\alpha \log \alpha, \alpha \log\log\log n\})$ implicit coloring algorithm of Christiansen et al.~\cite{Christiansen23}.} As the out-degree of each node is at most $d$, the potential can increase by at most $d$ when we add a node to $V^{*}\cup V^{colored}$. When a node $w$ gets a recursive call invoked on it, and thus added to $V^{*}\cup V^{colored}$, this reduces the potential by $6d$ and the subsequent processing of its outgoing arcs increases the potential by at most $d$. Hence, the recursive call on $w$ has a net effect of a potential reduction of at least $5d$. Thus, querying $i$ nodes for their colors has added at created at most $i/5$ recursive calls, besides the at most $i$ initial calls issued because of the query. Hence, the total number of nodes added to $V^*\cup V^{colored}$ is at most $6i/5$, which finishes the proof.
\end{proof}

\begin{proof}[Proof of \Cref{thm:amortized}]
The proof follows directly by combining two ingredients: (1) the orientation algorithm of Christiansen et al.~\cite{christiansen2022adaptive}, as summarized in \Cref{thm:orientation}, which dynamically keeps an $O(\alpha)$ orientation with update time $\poly(\log n)$. (2) the algorithm described above that, given an orientation with outdegree at most $d=O(\alpha)$, computes an implicit $O(d)$ coloring of the graph with amortized query time $O(d)$, as proven in \Cref{lem:amortized}. 
\end{proof}

As a side remark, we note that by combining \Cref{thm:amortized} with the random vertex partitioning idea as summarized in \Cref{thm:partition}, one gets a dynamic algorithm for $O(\alpha)$ implicit coloring with $\poly(\log n)$ update time and amortized $\poly(\log n)$ query time, against an oblivious adversary. This already improves over the algorithm of Christiansen et al.~\cite{Christiansen23} that computed an $O(\min\{\alpha \log \alpha, \alpha \log\log\log n\})$ implicit coloring with $\poly(\log n)$ update time and $\poly(\log n)$ amortized query time.

\subsection{An Intuitive Discussion of Our Worst-Case Algorithm}
Here, we discuss the intuition behind \Cref{thm:main}, which provides a randomized dynamic algorithm that computes an $O(\alpha)$ implicit coloring in $\poly(\log n)$ worst-case update and query times, against an oblivious adversary. In particular, we discuss a randomized algorithm with worst-case query time $\poly(\alpha \log n)$, which put together with the random partitioning (\cref{thm:partition}) implies a $\poly(\log n)$ query time. The actual algorithm and its analysis are presented in the next subsection.

\paragraph{Intuitive Discussion about Deamortization} The general framework of our worst-case algorithm is the same as the one described in the previous subsection, and we also use exactly the same procedure to color nodes of $V^*$. The only difference is in how the set $V^*$ is computed. Recall that we refer to the arcs in $A^p$ as processed arcs.

One reason why the algorithm above only gives an amortized guarantee is the following: It uses a fixed threshold on the number of processed incoming arcs in order to decide whether to invoke a recursion on an outgoing neighbor $w \notin V^* \cup V^{colored}$. This can lead to a large number of recursions being triggered by a single call to the algorithm computing $V^*$, e.g., if many of the nodes have their number of processed incoming arcs just barely below this threshold. How can we avoid this and obtain a worst-case guarantee?

A natural idea to prevent the issue is to use randomness to decide whether to trigger a recursion or not, with the hope that it is unlikely that we cause a lot of recursions together. A natural attempt in that regard would be to use uniform probabilities: perform a recursive call on an outgoing neighbor $w \notin V^* \cup V^{colored}$ with probability $p$, for some fixed $p$. Unfortunately, this approach does not work, as we explain next. If $p \geq \frac{10}{d}$, then a node with $d$ out-neighbors would perform up to $10$ recursive calls, in expectation, and it is easy to construct an example where the expected number of recursive calls on an initially uncolored graph is much more than $\poly(\alpha\log n)$. On the other hand, for any $p < \frac{100}{d}$, there is a positive constant probability that a node $w \notin V^* \cup V^{colored}$ has more than $1000d$ incoming arcs that have already been processed, and thus we would not be able to maintain the invariants. Intuitively, the issue with using a uniform probability is that it does not take into account how urgent it is to add $w$ to $V^*$, where the urgency increases with an increasing number of incoming arcs that have been processed. 

\paragraph{Our Worst-Case Algorithm}
Our worst-case algorithm triggers a recursion on an outgoing neighbor $w \notin V^* \cup V^{colored}$ with probability $\frac{1}{6d+1 - |A^p \cap IN(w)|}$ where $IN(w)$ denotes the set of arcs pointing towards $w$. In particular, the probability of invoking a recursion increases with the number of incoming arcs that have been processed. Thus, our algorithm can be seen as interpolating between the amortized approach that uses a fixed threshold on the number of processed arcs to decide whether to trigger a recursion and the approach that triggers a recursion with uniform probability.
Note that if $w$ has $6d$ incoming arcs that have been processed, then a recursion is triggered with probability $1$. Thus, computing $V^*$ in this manner still maintains the same invariants as before. The main difficulty is to show that the size of $V^*$ is small at all times, with high probability.

\paragraph{Analysis}
Our goal is to show that, in one color query, the number of recursive calls performed by the algorithm computing $V^*$ is $O(\log n)$, with high probability. Note that this guarantee does not hold if the set of previously colored nodes $V^{colored}$ is provided by an adversary, even if the adversary ensures that each node not in $V^{colored}$ has at most $6d$ in-neighbors in $V^{colored}$.

However, if one starts with an initially uncolored graph and queries the color of all the nodes in the graph, then with high probability the size of $V^*$ will be upper bounded by $O(\log n)$ throughout the execution of the algorithm, even if the queries are performed in an adaptive manner. 

\paragraph{Part I: A Simpler Case}
As a warm-up, let us first analyze the variant where we perform a recursive call with some fixed probability $p$, concretely $p = \frac{1}{6d}$. For this variant, it is strictly easier to show that the algorithm performs $O(\log n)$ recursive calls with high probability, and this guarantee even holds starting with an arbitrary subset of colored nodes. Assume that the algorithm performs more than $100\log(n)$ recursive calls and let $u$ be the initial input node that was queried for its color. Then, there exists a certificate $A^c \subseteq A$ with $|A^c| = \lceil100\log(n)\rceil$ such that for each arc $(v,w)$ in the certificate, $v$ performed a recursive call on $w$. Note that each such potential certificate $A^c$ has the property that it induces a directed tree rooted at $u$, where all arcs are oriented away from the root. Using this simple observation, the number of potential certificates can be upper bounded by $(4d)^{\lceil100\log(n)\rceil}$ using a standard counting argument. Moreover, the probability that some arbitrary subset $A' \subseteq A$ of size $\lceil100\log(n)\rceil$ is actually a certificate is at most $p^{\lceil100\log(n)\rceil}$, as each recursion is triggered independently with probability $p$. Thus, a simple union bound implies that the algorithm performs $O(\log n)$ recursive calls with probability $1 - \left(\frac{4d}{6d}\right)^{\lceil100\log(n)\rceil} = 1 - 1/\poly(n)$, as needed. We note that the argument above is not new; an analysis very similar to this was first used by Beck~\cite{beck1991algorithmic}, and variants of it can be found throughout the literature.

\paragraph{Part II: The Full Case} Now, let us analyze our actual algorithm. Recall that our algorithm performs a recursive call on a node $w \notin V^* \cup V^{colored}$ with probability $\frac{1}{6d + 1 - |A^p \cap IN(w)|}$.
The analysis is similar to the one above, with the only difference being how we show that a given set $A' \subseteq A$ is a certificate with probability at most $\left(\frac{1}{6d} \right)^{\lceil100\log(n)\rceil}$.
More precisely, consider the setting from above where no node is colored at the beginning and where all of the nodes are queried in an arbitrary order, where the order might even be chosen adaptively.
Then, for a fixed subset $A' \subseteq A$, it happens with probability at most $\left(\frac{1}{6d}\right)^{|A'|}$ that for every $(v,w) \in A'$, node $v$ invokes a recursive call on node $w$ while processing one of the queries.
The main intuition why this is the case is as follows: Fix some $i \in \{0,1,\ldots,6d-1\}$ and some node $w$. Then, the probability that $w$ gets called recursively at the time when it has exactly $i$ incoming arcs that have been processed is at most

\[\frac{6d-1}{6d}\cdot \frac{6d-2}{6d-1} \cdot \ldots \frac{6d+1-i}{6d+1-i +1} \cdot \frac{1}{6d+1-i} = \frac{1}{6d}.\]

This follows because for this event to happen, for every $j \in \{1,2,\ldots,i\}$, while processing the $j$-th arc incoming to $w$, a biased coin coming up heads with probability $\frac{1}{6d+1-j}$ was thrown to decide whether to invoke a recursion on $w$, but the coin came up tails, and while processing the incoming arc that triggered the recursion, a biased coin coming up heads with probability $\frac{1}{6d+1-i}$ was thrown, and the coin actually came up heads.

There are two issues that have not been handled. We mention these here, but leave the resolutions to the next subsection, where we discuss the algorithm and its actual analysis. (1) The argument above only gives that the probability is at most $\frac{1}{6d}$ that a given node $w$ is called recursively at the time it has $i$ incoming arcs that have been processed. This however does not directly imply that for a fixed arc $(v,w)$, with probability at most $\frac{1}{6d}$, node $v$ invokes a recursion on $w$. (2) Moreover, and even more importantly, for two arcs $(v_1,w_1)$ and $(v_2,w_2)$, the event that $v_1$ invokes a recursion on $w_1$ is not independent of the event that $v_2$ invokes a recursion on $w_2$. We see the details of how to handle these issues in the next subsection.

\subsection{Coloring with Worst-Case Query Time}
In this subsection, we describe our randomized algorithm that computes an $O(\alpha)$ implicit coloring in $\poly(\log n)$ update time and $\poly(\alpha \log n)$ worst-case query time. Putting this together with the known randomized partitioning idea, we get an $O(\alpha)$ arboricity in $\poly(\log n)$ worst-case update and query times. We assume here that we are given a $d$ out-degree orientation with $d = O(\alpha)$. We make no further assumptions on the out-degree orientation, and in particular, the orientation may include cycles.

\paragraph{Algorithm} The general framework of our worst-case algorithm and all the invariants are the same as the one described in the previous subsection, and we also use exactly the same procedure to color nodes of $V^*$. The only difference is in how the set $V^*$ is computed. To have a small worst-case complexity, it is critical that we now keep the property that $V^*$ is never too large (during one color query). For that, we change the rule for which vertices are added to $V^*$ recursively.

When we have a recursive call on node $v$ (including the first call caused by a color query on $v$), we do as follows. First, we add $v$ to $V^*$. Then, we process all arcs outgoing from $v$. For each arc $(v,w)$ outgoing from $v$, we first add $(v, w)$ to $A^{p}$. If $w$ is in $V^{colored}\cup V^{*}$, we are done with processing this arc. Suppose that $w$ is not in $V^{colored}\cup V^{*}$. Then, we first add $(v, w)$ to the list of processed arcs incoming to $w$. Then, we toss a coin with heads probability $\min\{\frac{1}{6d+1-|A^{p}\cap IN(w)|}, 1\}$, where $IN(w)$ indicates the set of arcs incoming to node $w$. If the coin came out heads, then we invoke a recursion on node $w$. If the coin was tails,  we do not recurse on $w$ at this point. We emphasize that in this latter case, we might decide to recurse on $w$ later when we process other arcs incoming to $w$. Note that the general scheme in this recursive algorithm is exactly the same as our amortized algorithm; it is just the decision whether to trigger a recursive procedure at $w$ that is different and decided randomly here, instead of a fixed deterministic threshold. However, once $|A^{p}\cap IN(w)|=6d$, the probability to trigger a recursive call on $w$ is $1$, and therefore, the correctness guarantees and invariants are the same as before. It is only the analysis for the size of $V^*$ that requires a new argument, as we provide next. 


\begin{lemma}\label{lem:randomized}
With each new query, the size of the set $V^*$ of vertices formed to be colored is at most $O(\log n)$, with high probability. Hence, the query time of the implicit coloring algorithm is at most $O(d\log n)$, with high probability.
\end{lemma}
\begin{proof}
Assume that the adversary queries the color of all the nodes, potentially in an adaptive manner. We show that  $|V^*| = O(\log n)$ with high probability during the whole execution of the algorithm. 

For the analysis, we keep track of a set $A^{r}\subseteq A^{p}$ of \textit{recursion-causing processed arcs}, which is initially empty after any edge update. When tossing a coin at the time of processing the arc $(v, w)$, if the coin came out heads and we triggered a recursion on $w$, then we add arc $(v, w)$ to $A^r$. On the other hand, if the coin was tails, we do not add the arc $(v, w)$ to $A^{r}$ (and we will never add it in the future, as the arc is processed at most once).

If the number of newly colored vertices is strictly more than $\lceil 100 \log(n)\rceil$ while answering the color query of some node $u$, then there exists some certificate $A^c$ such that $A^c \subseteq A^r$ after querying all the nodes and $A^c$ additionally satisfies that

\begin{enumerate}
    \item $|A^c| = \lceil 100 \log(n)\rceil$ and
    \item $A^c$ induces a tree rooted at $u$ with the arcs pointing away from $u$.
\end{enumerate}

The number of subsets $A^c \subseteq A$ satisfying the two properties above for some node $u \in V$ can be upper bounded by $n \cdot (4d)^{\lceil 100 \log(n)\rceil}$ using a simple and standard counting argument \cite{beps}.
We next show that for a fixed set $A' \subseteq A$ with $|A'| = \lceil 100 \log(n)\rceil$, the probability that $A' \subseteq A^r$ after all nodes have been queried is at most $\left(\frac{1}{6d}\right)^{\lceil 100 \log(n)\rceil}$. In particular, a union bound then allows us to conclude that, with probability at least 
\[1 - n \cdot \left(\frac{4d}{6d}\right)^{\lceil 100 \log(n)\rceil} \geq 1  - \frac{1}{n^{10}},\]
it holds that $|V^*| = O(\log n)$ throughout the processing of all the queries.

For an arc $a=(u, v)$, let us use the notation $head(a)=v$. In order to show that the probability that $A' \subseteq A^r$ after all queries have been processed is at most $\left(\frac{1}{6d}\right)^{\lceil 100 \log(n)\rceil}$, we show a more general claim by induction. Namely, during any point while answering the queries, we can upper bound the probability that $A' \subseteq A^r$ at the end of the algorithm by

$$f(A^p,A^r) :=
\begin{cases}
0 &\text{If $A' \cap A^p \not \subseteq A^r$} \\
0 & \text{If there exists $a' \in A'$ and $a^r \in A^r \setminus \{a'\}$ and $head(a') = head(a^r)$}\\
\prod\limits_{a' \in A' \setminus A^p} \frac{1}{6d - |IN(head(a')) \cap A^p|} &\text{Otherwise.}
\end{cases}$$

Note, if $A'\cap A^p \not \subseteq A^r$, then there exists an arc $a'$ in $A'$ that was processed but which did not trigger a recursion, and thus will also not trigger a recursion in the future. Also, note that for each node there exists at most one incoming arc which triggered a recursion.
Therefore, if there exists some arc $a'$ in the certificate $A'$ pointing to the same node as some other arc $a^r$ that triggered a recursion, then $a'$ has not and will never trigger a recursion in the future. Thus, if we are in the first two cases, and therefore $f(A^p,A^r) = 0$, we indeed get the property that $A' \not \subseteq A^r$ at the end of the algorithm. 
Moreover, note that the statement is indeed more general as we have $A^p = A^r = \emptyset$ at the beginning and $f(\emptyset,\emptyset) = \left(\frac{1}{6d}\right)^{\lceil 100 \log(n) \rceil}$. Also, note that if $A'$ has two arcs pointing to the same vertex, then it cannot happen that both of them end up in $A^r$, and therefore we can assume from now on that for every vertex, there is at most one arc in $A^c$ pointing to it.

For the proof below, we define $q(A^p,a') = \frac{1}{6d - |IN(head(a')) \cap A^p|}$ and therefore

\[\prod\limits_{a' \in A' \setminus A^p} \frac{1}{6d - |IN(head(a')) \cap A^p|} = \prod\limits_{a' \in A' \setminus A^p} q(A^p,a').\]

We prove the statement by reverse induction on $|A^p|$. As our base case, we consider the time when the algorithm has answered all the queries. To prove the base case, we have to show that $A' \subseteq A^r$ implies $f(A^p,A^r) = 1$. First, note that $A' \subseteq A^r$ trivially implies $A' \cap A^p \subseteq A^r$, thus we are not in the first case.
Now, consider some arc $a^r \in A^r$ that triggered a recursion. Then, for any other arc $a$ pointing to the same node, we know that $a$ did not trigger a recursion.
Therefore, $a \notin A^r$ and as $A'\subseteq A^r$, we also get $a \notin A'$. \\
Hence, we are in the third case. As $A^r \subseteq A^p$ and $A' \subseteq A^r$, we have $A' \setminus A^p = \emptyset$ and thus we indeed have $f(A^p,A^r) = 1$.

For our induction step, consider some point in time when the algorithm is about to add some arc $a$ to $A^p$. We first consider the more interesting case that the head of $a$ is not already in $V^* \cup V^{colored}$ and therefore the algorithm is about to flip a coin that comes up heads with probability $\frac{1}{6d + 1 - |IN(head(a)) \cap (A^p \cup \{a\})|} = \frac{1}{6d - |IN(head(a)) \cap A^p|} = q(A^p,a)$, and add $a$ to $A^r$ if the coin comes up heads. Later, we discuss the case that the head of $a$ is already in $V^* \cup V^{colored}$.

Using the induction hypothesis, we can conclude that $A' \subseteq A^r$ at the end of the algorithm with probability at most 

\[q(A^p,a) \cdot f(A^p \cup \{a\},A^r \cup \{a\}) + (1-q(A^p,a))f(A^p \cup \{a\},A^r).\]

Thus, it suffices to show that
\[q(A^p,a) \cdot f(A^p \cup \{a\},A^r \cup \{a\}) + (1-q(A^p,a))f(A^p \cup \{a\},A^r) \leq f(A^p,A^r).\]
If one of the first two cases apply, i.e., if $A' \cap A^p \not \subseteq A^r$ or there exists $a' \in A'$ and $a^r \in A^r \setminus \{a'\}$ with $head(a') = head(a^r)$, then $f(A^p \cup \{a\},A^r \cup \{a\}) = f(A^p \cup \{a\},A^r) = 0$. 

Thus, it remains to consider the case $A' \cap A^p \subseteq A^r$ and there does not exist $a' \in A'$ and $a^r \in A^r \setminus \{a'\}$ with $head(a') = head(a^r)$. We distinguish between the following three cases:

\begin{enumerate}
    \item $a \in A'$ 
    \item $head(a) \neq head(a')$ for every $a' \in A'$ and
    \item $a \notin A'$ and there exists an arc $a'' \in A'$ with $head(a) = head(a'')$.
\end{enumerate}

First, consider the case $a \in A'$. In that case, we have $f(A^p \cup \{a\},A^r) = 0$. Thus, it suffices to show that $q(A^p,a) \cdot f(A^p \cup \{a\},A^r \cup \{a\}) = f(A^p,A^r)$.
It holds that $f(A^p \cup \{a\},A^r \cup \{a\}) = \prod\limits_{a' \in A' \setminus (A^p \cup \{a\})} q(A^p \cup \{a\},a')$. Note that if $head(a) \neq head(a')$, then $q(A^p \cup \{a\},a') = q(A^p,a')$. Then, 

\[f(A^p \cup \{a\},A^r \cup \{a\}) = \prod\limits_{a' \in A' \setminus (A^p \cup \{a\})} q(A^p,a').\]

As $a \in A' \setminus A^p$ but trivially $a \notin A' \setminus (A^p \cup \{a\})$, we can conclude that

\[f(A^p,A^r) =  \prod\limits_{a' \in A' \setminus A^p} q(A^p,a') = q(A^p,a) \cdot \prod\limits_{a' \in A' \setminus (A^p \cup \{a\})} q(A^p,a') = q(A^p,a) \cdot f(A^p \cup \{a\},A^r \cup \{a\}),\]

as needed.

Next, we consider the case that $head(a) \neq head(a')$ for every $a' \in A'$. This is the simplest case, as it is easy to show that in this case it holds that $f(A^p,A^r) = f(A^p \cup \{a\},A^r) = f(A^p \cup \{a\},A^r \cup \{a\})$.

Thus, it remains to consider the case that $a \notin A'$ and there exists an arc $a''\in A'$ with $head(a) = head(a'')$.
In particular, we have $a'' \notin A^r$ and as we assume that $A' \cap A^p \subseteq A^r$, we also get that $a'' \notin A^p$.
First, note that $f(A^p \cup \{a\},A^r \cup \{a\}) = 0$.
Thus, it is sufficient for us to show that $(1 - q(A^p,a)) \cdot f(A^p \cup \{a\},A^r) = f(A^p,A^r)$. We have

\begin{align*}
    q(A^p,a'') &= \frac{1}{6d - |IN(head(a'')) \cap A^p|} \\
                          &= \frac{1}{6d - |IN(head(a'')) \cap A^p| - 1} \cdot \frac{6d - |IN(head(a'')) \cap A^p| - 1}{6d - |IN(head(a'')) \cap A^p|} \\
                          &=  \frac{1}{6d - |IN(head(a'')) \cap A^p| - 1} \cdot \left(1 - \frac{1}{6d - |IN(head(a'')) \cap A^p|} \right) \\
                          &=  \frac{1}{6d - |IN(head(a'')) \cap (A^p \cup \{a\})|} \cdot \left(1 - \frac{1}{6d - |IN(head(a)) \cap A^p|} \right) \\
                          &= q(A^p \cup \{a\},a'') \cdot (1 - q(A^p,a)).
\end{align*}

Since $q(A^p \cup \{a\},a') = q(A^p,a')$ for every $a' \in A' \setminus \{a''\}$, we therefore get

\begin{align*}
f(A^p,A^r) &= \prod\limits_{a' \in A' \setminus A^p} q(A^p,a') \\
&= \frac{q(A^p,a'')}{q(A^p \cup \{a\},a'')} \cdot \prod\limits_{a' \in A' \setminus A^p} q(A^p \cup \{a\},a') \\
&= (1-q(A^p,a)) \prod\limits_{a' \in A' \setminus A^p} q(A^p \cup \{a\},a') \\
&= (1-q(A^p,a)) \prod\limits_{a' \in A' \setminus (A^p \cup \{a\})} q(A^p \cup \{a\},a') \\
&= (1-q(A^p,a)) \cdot f(A^p \cup \{a\},A^r).
\end{align*}

Finally, it remains to consider the case that $head(a) \in V^* \cup V^{colored}$, in which case the algorithm does not recurse on $a$. Thus, we can use the induction hypothesis to conclude that the probability is at most $f(A^p \cup \{a\},A^r)$ that  $A' \subseteq A^r$. Note that $f(A^p \cup \{a\},A^r) \leq f(A^p,A^r)$ unless there exists an arc $a' \in A' \setminus A^r$ with $head(a') = head(a)$. However, then $a'$ is pointing to a node in $V^* \cup V^{colored}$ and therefore $a' \notin A^r$ after all queries have been processed.
\end{proof}

\begin{proof}[Proof of \Cref{thm:main}]
By \Cref{thm:partition}, to get $O(\alpha)$ coloring in $\poly(\log n)$ update time and $\poly(\log n)$ query time, it suffices to have a dynamic algorithm that maintains an $O(\alpha)$ implicit coloring of the graph in $\poly(\log n)$ update time and $\poly(\alpha \log n)$ query time. The latter follows directly from two ingredients: (1) The orientation algorithm of Christiansen et al.~\cite{christiansen2022adaptive}, as summarized in \Cref{thm:orientation}, which dynamically keeps an $O(\alpha)$ orientation with update time $\poly(\log n)$. (2) The algorithm described above that, given an orientation with outdegree at most $d=O(\alpha)$, computes an implicit $O(d)$ coloring of the graph with worst-case query time $O(d\log n)$, as proven in \Cref{lem:randomized}. 
\end{proof}
\subsection*{Acknowledgement}
We thank Aleksander Christiansen, Krzysztof Nowicki, and Eva Rotenberg for helpful feedback about a preliminary draft of this work and connections to theirs~\cite{Christiansen23}.
\bibliographystyle{alpha}
\bibliography{sample}

\end{document}